\documentclass[10pt,journal,compsoc]{IEEEtran}
\usepackage{graphicx,amsmath,amsthm}
\usepackage{balance}  % for  \balance command ON LAST PAGE  (only there!)
\usepackage[ruled,vlined,linesnumbered,commentsnumbered]{algorithm2e}
\usepackage{color}
\usepackage{booktabs}
\usepackage{caption}
\usepackage{url}

\newcommand{\remove}[1]{}

\newcommand{\maxbicliques}{g}
\newcommand{\mbc}{{\tt DynamicBC}}

\newtheorem{theorem}{Theorem}

\newtheorem{lemma}{Lemma}

\newtheorem{definition}{Definition}

\newtheorem{observation}{Observation}

\DontPrintSemicolon
\SetAlgoLined

\newcommand{\bicliques}{{\cal BC}}
\newcommand{\mbe}{{\tt MineLMBC}}
\newcommand{\csnewb}{{\tt NewBC}}
\newcommand{\cssubb}{{\tt SubBC}}

\begin{document}

\title{A Change-Sensitive Algorithm for Maintaining Maximal Bicliques in a Dynamic Bipartite Graph}

\author{Apurba~Das,
        Srikanta~Tirthapura,~\IEEEmembership{Senior Member,~IEEE,}% <-this % stops a space
\IEEEcompsocitemizethanks{\IEEEcompsocthanksitem Das and Tirthapura are with the Department of Electrical
and Computer Engineering, Iowa State University, Ames, IA 50011.\protect\\
E-mail: {\tt \{adas,snt\}@iastate.edu}. \protect\\
The authors were partially supported through the NSF grant 1527541.
}
}

\IEEEtitleabstractindextext{%
\begin{abstract}
We consider the maintenance of maximal bicliques from a dynamic bipartite graph that changes over time due to the addition or deletion of edges.  When the set of edges in a graph changes, we are interested in knowing the change in the set of maximal bicliques (the ``change"), rather than in knowing the set of maximal bicliques that remain unaffected. The challenge in an efficient algorithm is to enumerate the change without explicitly enumerating the set of all maximal bicliques. In this work, we present (1) near-tight bounds on the magnitude of change in the set of maximal bicliques of a graph, due to a change in the edge set (2) a ``change-sensitive" algorithm for enumerating the change in the set of maximal bicliques, whose time complexity is proportional to the magnitude of change that actually occurred in the set of maximal bicliques in the graph. To our knowledge, these are the first algorithms for enumerating maximal bicliques in a dynamic graph, with such provable performance guarantees. Our algorithms are easy to implement, and experimental results show that their performance exceeds that of current baseline implementations by orders of magnitude. 
%However, the challenge is to find new patterns as the graph evolves over time whithout enumerating patterns in the entire graph. We address this challenge by designing efficient algorithms that can find new dense structures by exploring local regin of the graph where the changes occur.
%In this work, we design algorithms for maintaining maximal bicliques from bipartite graph when the graph changes due to addition of edges. We are the first to propose such an algorithm whose time complexity is proportional to the magnitude of the change in the set of bicliques. We also present nearly tight bounds on the size of changes in the set of maximal bicliques when the graph is updated.
\end{abstract}

%\begin{IEEEkeywords}
%Computer Society, IEEE, IEEEtran, journal, \LaTeX, paper, template.
%\end{IEEEkeywords}
}

\maketitle

%\doublespacing
%\input{algorithm-new}
%---------------------------
\section{Introduction}
\label{sec:intro}

%---------------------------
Graphs are ubiquitous in representing linked data in many domains such as in social network analysis, computational biology, and web search. Often, these networks are dynamic,  where new connections are being added and old connections are being removed. The area of {\em dynamic graph mining} focuses on efficient methods for finding and maintaining significant patterns in a dynamic graph. In this work we focus on the maintenance of dense subgraphs within a dynamic graph.

% in recent years that gives birth to a field called \textit{dynamic graph mining}. A fundamental goal is to find interesting subgraphs and to maintain those subgraphs through the evolution of the graph over time. These interesting subgraphs are usually dense in the sense that the vertices are tightly interconnected.

Our work is motivated by many applications that require the maintenance of dense substructures from a dynamic graph. Angel et al.~\cite{AKS+13}, propose an algorithm for  identifying breaking news stories in real-time through dense subgraph mining from an evolving graph, defined on the co-occurrence of entities within messages in an online social network. \cite{JS+07} present methods for detecting communities among users in a microblogging platform through identifying dense structures in an evolving network representing connections among users. A sample of other applications of dense subgraph mining in networks include identification of communities in a social network~\cite{hanneman-socialnw,LSZL2011}, identification of web communities~\cite{GKT05,RH2005,KR+99}, phylogenetic tree construction~\cite{DABMMS2004,SDREL2003,YBE2005}, communities in bipartite networks~\cite{LSH08}, genome analysis~\cite{NK08}, and closed itemset mining~\cite{VMG08,LL+05}.

We consider the fundamental problem of maintaining maximal bicliques in a bipartite graph that is changing due to the addition or deletion of edges. Let $G = (L, R, E)$ be a simple undirected bipartite graph with its vertex set partitioned into $L$, $R$, and edge set $E \subseteq L\times R$. A biclique in $G$ is a bipartition $B = (X, Y)$, $X\subseteq L$, $Y\subseteq R$ such that each vertex in $X$ is connected to each vertex in $Y$. A biclique $B$ is called a maximal biclique if there is no other biclique $B'$ such that $B$ is a proper subgraph of $B'$. Let $\bicliques(G)$ denote the set of all maximal bicliques in $G$.

Suppose that, starting from bipartite graph $G_1 = (L, R, E)$, the state of the graph changes to $G_2 = (L, R, E\cup H)$ due to the addition of a set of new edges $H$. Let $\Upsilon^{new}(G_1, G_2) = \bicliques(G_2)\setminus\bicliques(G_1)$ denote the set of new maximal bicliques that arise in $G_2$ that were not present in $G_1$ and $\Upsilon^{del}(G_1, G_2) = \bicliques(G_1)\setminus\bicliques(G_2)$ denote the set of maximal bicliques in $G_1$ that are no longer maximal bicliques in $G_2$ (henceforth called as subsumed bicliques). See Fig.~\ref{fig:new-subsumed} for an example. Let $\Upsilon(G_1, G_2) = \Upsilon^{new}(G_1, G_2)\cup \Upsilon^{del}(G_1, G_2)$ denote the symmetric difference of $\bicliques(G_1)$ and $\bicliques(G_2)$. We ask the following questions:

(1)~How large can be the size of $\Upsilon(G_1, G_2)$? In particular, can a small change in the set of edges cause a large change in the set of maximal bicliques in the graph?

(2)~How can we compute $\Upsilon(G_1, G_2)$ efficiently? Can we quickly compute $\Upsilon(G_1, G_2)$ when $|\Upsilon(G_1, G_2)|$ is small? In short, can we design  \textit{change-sensitive algorithms} for enumerating elements of $\Upsilon(G_1, G_2)$, whose time complexity is proportional to the size of change, $|\Upsilon(G_1, G_2)|$?

%{\color{red} useful to have a picture here showing a bipartite graph, a set of edges $H$ added, and new bicliques and subsumed bicliques.}
%-------------------------------------------------------------------------------------
\begin{figure}[!t]
\centering
\begin{tabular}{c}
%\hspace{-2mm}
%\includegraphics[width=.45\textwidth]{pic2.png}\\
%\includegraphics[width=.45\textwidth]{bpic1.png}\\
\includegraphics[width=3in]{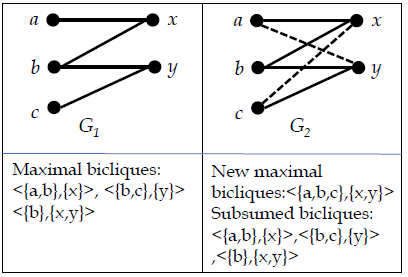}\\
\end{tabular}
\caption{Change in maximal bicliques when the graph changes from $G_1$ to $G_2$ due to the addition of edge set $H = \{\{a,y\}, \{c,x\}\}$. Note that each maximal biclique of $G_1$ is subsumed by a larger maximal biclique in $G_2$, and there is one new maximal biclique in $G_2$.}
\label{fig:new-subsumed}
%fig7
\end{figure}
%-------------------------------------------------------------------------------------
%-------------------------------
\subsection{Contributions}
%-------------------------------

{\bf Magnitude of Change:}
Let $\maxbicliques(n)$ denote the maximum number of maximal bicliques possible in an $n$ vertex bipartite graph. A result due to Prisner~\cite{P00} shows that $\maxbicliques(n) \leq 2^{n/2}$, where equality occurs when $n$ is even. We show that the change in the number of maximal bicliques when a single edge is added to the graph can be as large as $3\maxbicliques(n-2) \approx 1.5\times2^{n/2}$, which is exponential in the number of vertices in the graph. This shows that the addition of even a single edge to the graph can lead to a large change in the set of maximal bicliques in the graph. We further show that this bound is tight for the case of the addition of a single edge -- the largest possible change in the set of maximal bicliques upon adding a single edge is $3\maxbicliques(n-2)$. For the case when more edges can be added to the graph, it is easy to see that the maximum possible change is no larger than $2\maxbicliques(n)$.\\

%Let $\maxbicliques(n)$ denote the maximum number of maximal bicliques possible in an $n$ vertex bipartite graph. A result due to Prisner~\cite{P00} shows that {\color{red}$\maxbicliques(n) \leq 2^{n/2}$. $\maxbicliques(n) = 2^{n/2}$ when $n$ is even.} \remove{\footnote{\color{red}give the exact expression here}}. We show that the change in the number of bicliques when a single edge is added to the graph can be as large as $3\maxbicliques(n-2) \approx 1.5\times2^{n/2}$, which exponential in the number of edges in the graph. We further show that this bound is tight -- this is the largest possible change in the set of maximal bicliques upon adding a single edge. For the case when more edges can be added to the graph, we show that the maximum possible change can be {\color{red} no largher than $2\maxbicliques(n) \approx 2\times2^{n/2}$ which is a trivial upper bound. We also show that the change is at least $1.5\maxbicliques(n)$ when $n$ is even, and $\sqrt{2}\maxbicliques(n)$ when $n$ is odd.} \remove{{\color{red} give the expression here.}}

%tight upper bound on the size of change in the set of maximal bicliques when a single edge is added to the graph. Specifically, the bound is $3\maxbicliques(n-2) = 1.5\times2^{n/2}$ for an $n$ vertex bipartite graph when $n$ is even. So change in maximal bicliques due to just a single edge addition can be large.

\noindent{\bf Enumeration Algorithm:}
From our analysis, it is clear that the magnitude of change in the set of maximal bicliques in the graph can be as large as exponential in $n$ in the worst case. On the flip side, the magnitude of change can be as small as $1$ -- for example, consider the case when a newly arriving edge connects two isolated vertices in the graph. Thus, there is a wide range of values the magnitude of change can take. When the magnitude of change is very large, an algorithm that enumerates the change must inevitably pay a large cost, if only to enumerate the change. On the other hand, when the magnitude of change is small, it will ideally pay a smaller cost. This motivates our search for a {\bf change-sensitive algorithm} whose computational cost for enumerating the change is proportional to the magnitude of the change in the set of maximal bicliques.

We present a change-sensitive algorithm, $\mbc$, for enumerating the new maximal bicliques and subsumed maximal bicliques, when a set of new edges $H$ are added to the bipartite graph $G$. The algorithm $\mbc$ has two parts, $\csnewb$, for enumerating new maximal bicliques, and $\cssubb$, for enumerating subsumed maximal bicliques. When a batch of new edges $H$ of size $\rho$ is added to the graph, the time complexity of $\csnewb$ for enumerating $\Upsilon^{new}$, the set of new maximal bicliques, is $O(\Delta^2\rho|\Upsilon^{new}|)$ where $\Delta$ is the maximum degree of the graph after update. The time complexity of $\cssubb$ for enumerating $\Upsilon^{del}$, the set of subsumed bicliques, is $O(2^{\rho}|\Upsilon^{new}|)$. To the best of our knowledge, these are the first provably change-sensitive algorithms for maintaining maximal bicliques in a dynamic graph.

{\bf Experimental Evaluation:}
We present an empirical evaluation of our algorithms on real bipartite graphs with million of nodes. Our results shows that the performance of our algorithms are orders of magnitude faster than current approaches. For example, on the 
{\tt actor-movie-1} graph with $640$K vertices and $1.4$M edges, our algorithm 
took about 30 milliseconds for computing the change due to the addition of a batch of 100 edges, while the baseline algorithm took more than 30 minutes.

%{\tt dblp-author-1} graph with $6.8$M vertices and $8$M edges, our algorithm took around $20$ ms. for computing the change in the set of maximal bicliques due to the addition of a batch of $100$ edges, while the baseline algorithm took more than $20$ hours.

%--------------------------
\subsection{Related Work}
\label{subsec:rel}
%--------------------------
%\textbf{Maximal Bi-Clique enumeration (MBE) from static graph.} 
\textbf{Maximal Biclique enumeration (MBE) on a static graph:} There has been substantial prior work on enumerating maximal bicliques from a static graph. Alexe et al.~\cite{AA+04} propose an algorithm for MBE from a static graph based on the consensus method, whose time complexity is proportional to the size of the output (number of maximal bicliques in the graph) - termed as \textit{output-sensitive algorithm}. Liu et al.~\cite{LSL06} propose an algorithm for MBE based on depth-first-search (DFS). Damaschke~\cite{D14} propose an algorithm for bipartite graphs with a skewed degree distribution. G{\'e}ly et al.~\cite{GNS09} propose an algorithm for MBE through a reduction to  maximal clique enumeration (MCE).  However, in their work, the number of edges in the graph used for enumeration increases significantly compared to the original graph. Makino \& Uno~\cite{MU04} propose an algorithm for MBE based on matrix multiplication, which provides the current best time complexity for dense graphs. Eppstein~\cite{E94} proposes a linear time algorithm for MBE when the input graph has bounded arboricity. Other works on sequential algorithm for MBE on a static graph include~\cite{DDS05, DDS07}. \cite{MT14,XC+14,SMT14} present parallel algorithms for MBE and MCE for the MapReduce framework. \cite{LL+05} show a correspondence between closed itemsets in a transactional database and maximal cliques in an appropriately defined graph.

%,MT16

%DFS based algorithm in connection with closed pattern mining from transaction database has been proposed in~\cite{LL+07}.
%Correspondence between enumerating maximal bicliques and enumerating frequent closed itemsets has been shown in~\cite{LL+05} where the set of items and the set of transactions containing those items form a biclique. The authors prove that for each maximal biclique, there exists a distinct pair of closed itemset where each itemset is a bipartition of the biclique and thus algorithm for mining closed itemset can be used for mining maximal bicliques. But, the problem with this algorithm for maximal biclique enumeration are the following: (1) each maximal biclique will be enumerated twice which is wasteful in terms of computation time and (2) a single biparition of each maximal biclique will be reported by the algorithm; so post processing is required for computing the other bipartition.

%However, maximal biclique enumeration problem is different from itemset mining problem in a sense that itemset mining problem does not enumerate all maximal bicliques. This is because, there exists size constraint on itemsets and all maximal bicliques (corresponding itemsets) are not actually enumerated henceforth. Another issue is that each maximal biclique is enumerated twice if we consider itemset mining algorithm for biclique enumeration~\cite{LSL06}. 

\textbf{Dense Structures from Dynamic Graphs:} There have been some prior works related to maintenance of dense structures similar to maximal bicliques in dynamic graphs. Kumar et al.~\cite{KR+99} define $(i,j)$-core which is a biclique with $i$ vertices in one partition and $j$ vertices in another partition. In their work, the authors propose a dynamic algorithm for extracting non-overlapping maximal set of $(i,j)$-cores for interesting communities. \cite{SG+13,LYM14,HC+14} present methods for maintaining $k$-cores and $k$-trusses in a dynamic graph, and~\cite{DST16} present algorithms for maintaining maximal cliques in a dynamic graph.

{\bf Roadmap:} The remaining section are organized as follows. We present definitions and preliminaries in Section~\ref{prelim}. Then we describe our algorithms in Section~\ref{csalgo}, results on the size of change in the set of maximal bicliques in Section~\ref{moc}, and experimental results in Section~\ref{expt}.

%--------------------------------------------------------
\section{Preliminaries}
\label{prelim}
%--------------------------------------------------------

Let $V(G)$ denote the set of vertices of $G$ and $E(G)$ the set of edges in $G$. Let $n$ and $m$ denote the number of vertices and number of edges in $G$ respectively. 
%For each $v\in V(G)$, let $deg(v)$ denote the degree of $v$. 
Let $\Gamma_{G}(u)$ denote the set of vertices adjacent to vertex $u$ in $G$. If the graph $G$ is clear from the context, we use $\Gamma(u)$ to mean $\Gamma_{G}(u)$. For an edge $e = (u,v)\in E(G)$, let $G-e$ denote the graph after deleting $e\in E(G)$ from $G$ and $G+e$ denote the graph after adding $e\notin E(G)$ to $G$. For a set of edges $H$, let $G+H$ ($G-H$) denote the graph obtained after adding (deleting) $H$ to (from) $E(G)$. Similarly, for a vertex $v\notin V(G)$, let $G+v$ denote the graph after adding $v$ to $G$ and for a vertex $v\in V(G)$, let $G-v$ denote the graph after deleting $v$ and all its adjacent edges from $E(G)$. Let $\Delta(G)$ denote the maximum degree of a vertex in $G$ and $\delta(G)$ the minimum degree of a vertex in $G$.

%-------------------------------------------------------------------------------------
\begin{figure}[!t]
\centering
\captionsetup{justification=centering}
\begin{tabular}{c}
%\hspace{-2mm}
%\includegraphics[width=.45\textwidth]{pic2.png}\\
\includegraphics[width=.25\textwidth]{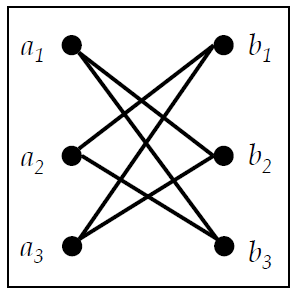}\\
\end{tabular}
\caption{Cocktail-party graph on 6 vertices $CP(3)$}
\label{fig:cocktail-party-graph}
%fig7
\end{figure}
%-------------------------------------------------------------------------------------

%\begin{definition}[Output-Sensitive Algorithm]
%An algorithm is called output-sensitive if its time complexity is proportional to the size of the output plus size of input.
%\end{definition}

\begin{definition}[Change-Sensitive Algorithm]
An algorithm for a dynamic graph stream is called change-sensitive if its time complexity of enumerating the change in a graph property is proportional to the magnitude of change.
\end{definition}

{\bf Results for a static graph.} 
In~\cite{P00}, Prisner presented the following result on the  number of maximal bicliques in a bipartite graph with $n$ vertices. 

%-------------------------------------------------------------------------------------
\begin{theorem}[Theorem 2.1~\cite{P00}]\label{thm1}
Every bipartite graph with $n$ vertices contains at most $2^{\frac{n}{2}}\approx 1.41^{n}$ maximal bicliques, and the only extremal (maximal) bipartite graphs are the graphs $CP(k)$.
\end{theorem}
%-------------------------------------------------------------------------------------

Here, $CP(k)$ denotes the \textit{cocktail-party} graph which is a bipartite graph with $k$ vertices in each partition where $V(CP(k)) = \{a_1, a_2, \ldots, a_k, b_1, b_2, \ldots, b_k\}$ and $E(CP(k)) = \{(a_i, b_p) : i \neq p\}$~\cite{P00}. See Figure~\ref{fig:cocktail-party-graph} for an example.

As a subroutine, we use  an algorithm for enumerating maximal bicliques from a static undirected graph, whose runtime is proportional to the number of maximal bicliques. There are a few algorithms of this kind~\cite{AA+04, LSL06, ZP+14}. We use the following result due to Liu et al.~\cite{LSL06} as it provides best possible time and space complexity.

%-------------------------------------------------------------------------------------
\begin{theorem}[Liu et al., \cite{LSL06}]\label{thm2}
For a graph $G$ with $n$ vertices, $m$ edges, maximum degree $\Delta$, and number of maximal bicliques $\mu$, there is an algorithm $\mbe$ for enumerating maximal bicliques in $G$ with time complexity $O(n\Delta\mu)$ and space complexity $O(m+\Delta^2)$.
\end{theorem}
%-------------------------------------------------------------------------------------

$\mbe$ is depth-first-search (DFS) based algorithm for enumerating maximal bicliques of a static graph $G = (V,E)$. It takes as input the graph $G$ and the size threshold $s$. The algorithm enumerates all maximal bicliques of $G$ with size of each partition at least $s$. Clearly, by setting $s=1$, the algorithm enumerates all maximal bicliques of $G$.

\section{Change-Sensitive Algorithm for Maximal Bicliques}
\label{csalgo}
%----------------------------------------------------------------------------------------
In this section, we present a change-sensitive algorithm $\mbc$ for enumerating the change in the set of maximal bicliques. The algorithm has two parts : (1)~Algorithm $\csnewb$ for enumerating new maximal bicliques, and (2)~Algorithm $\cssubb$ for enumerating subsumed bicliques. For graph $G$ and set of edges $H$, we use $\Upsilon^{new}$ to mean $\Upsilon^{new}(G,G+H)$,  and $\Upsilon^{del}$ to mean $\Upsilon^{del}(G, G+H)$.

%---------------------------- 
\begin{algorithm}
\DontPrintSemicolon
\caption{$\mbc(G,H, \bicliques(G))$}
\label{algo:mbc}
\KwIn{$G$ - Input bipartite graph, $H$ - Edges being added to $G$, $\bicliques(G)$}
\KwOut{$\Upsilon$ : the union of set of new maximal bicliques and subsumed bicliques}
$\Upsilon^{new}\gets \csnewb(G,H)$\;
$\Upsilon^{del}\gets \cssubb(G, H, \bicliques(G), \Upsilon^{new})$\;
$\Upsilon\gets \Upsilon^{new}\cup\Upsilon^{del}$\;
\end{algorithm}
%----------------------------

We first present Algorithm \csnewb~ for enumerating new cliques in Section~\ref{new_bicliques}, and Algorithm~\csnewb~ for enumerating subsumed cliques in Section~\ref{sub_bicliques}. The main result on the time complexity of $\mbc$ is summarized in the following theorem.

%----------------------------
\begin{theorem}\label{thm:main}
$\mbc$ is a change-sensitive algorithm for enumerating the change in the set of maximal bicliques, with time complexity $O(\Delta^{2}\rho|\Upsilon^{new}| + 2^{\rho}|\Upsilon^{new}|)$ where $\Delta$ is the maximum degree of a vertex in $G+H$ and $\rho$ is the size of $H$.
\end{theorem}
%----------------------------

%The proof of the theorem has two parts. In Theorem~\ref{thm:newbc-1} we prove that time complexity of $\csnewb$ is $O(\Delta^{2}\rho|\Upsilon^{new}|)$ and in Theorem~\ref{thm:subbc-1} we prove that time complexity of $\cssubb$ is $O(2^{\rho}|\Upsilon^{new}|)$.
%---------------------------------------------------------------------------------------
\subsection{Enumerating New Maximal Bicliques}
\label{new_bicliques}
%----------------------------------------------------------------------------------------
%-------------------------------------------------------------------------------------
\begin{figure}[h]
\centering
\begin{tabular}{c}
%\hspace{-2mm}
%\includegraphics[width=.45\textwidth]{pic2.png}\\
%\includegraphics[width=.45\textwidth]{bpic5.png}\\
\includegraphics[width=2.5in]{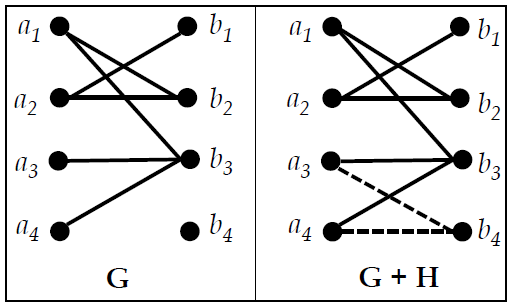}\\
\end{tabular}
\caption{The original graph $G$ has $4$ maximal bicliques. When new edges in $H$ (in dotted line) are added to $G$, all maximal bicliques in $G$ remain maximal in $G+H$ and only one maximal biclique is newly formed ($<\{a_3,a_4\},\{b_3,b_4\}>$).}
\label{fig:construction-1}
%fig7
\end{figure}
%-------------------------------------------------------------------------------------

Let $G'$ denote the graph $G+H$. A baseline algorithm for enumerating new maximal bicliques in $G'$ is to (1)~enumerate all maximal bicliques in $G$, (2)~enumerate all maximal bicliques in $G'$ both using an output-sensitive algorithm such as~\cite{LSL06}, and then (3)~compute $\bicliques(G')\setminus\bicliques(G)$. However, this is not change-sensitive, since we need to compute all maximal bicliques of $G'$ each time, but it is possible that most of the maximal bicliques in $G'$ are not new. For example, see Fig.~\ref{fig:construction-1}. We next present an approach that overcomes this difficulty. 

%In our algorithm, we construct subgraphs of $G'$ that are built starting from the newly inserted edges in $H$. These subgraphs are constructed so that the set of all new maximal bicliques in $G'$ that contain edge $e$ can be retrieved by enumerating all maximal bicliques in the subgraph constructed for edge $e$. By iterating this over all edges $e \in H$, we get an algorithm for enumerating all new maximal bicliques in $G'$.

%The time complexity of the  resulting algorithm depends on the number of new maximal bicliques in the updated graph due to the insertion of a set of new edges.

For each new edge $e\in H$, let $\bicliques'(e)$ denote the set of maximal bicliques in $G'$ containing edge $e$. 

%-------------------------------------------------------------------------------------
\begin{lemma}\label{lem:new1}
$\Upsilon^{new} = \cup_{e\in H}\bicliques'(e)$.
\end{lemma}
%-------------------------------------------------------------------------------------

%-------------------------------------------------------------------------------------
\begin{proof}
Each biclique in $\Upsilon^{new}$ must contain at least one edge from $H$. To see this, consider a biclique $b \in \Upsilon^{new}$. If $b$ did not contain an edge from $H$, then $b$ is also a maximal biclique in $G$, and hence cannot belong to $\Upsilon^{new}$. Hence, $b\in \bicliques'(e)$ for some edge $e\in H$, and $b\in\cup_{e\in H}\bicliques'(e)$. This shows that $\Upsilon^{new}\subseteq\cup_{e\in H}B'(e)$.

Next consider a biclique $b\in\cup_{e\in H}\bicliques'(e)$. It must be the case that $b\in \bicliques'(h)$ for some $h$ in $H$. Thus $b$ is a maximal biclique in $G+H$, and $b$ contains edge $h\in H$ and $b$ cannot be a biclique in $G$. Thus $b\in\Upsilon^{new}$. This shows that $\cup_{e\in H}\bicliques'(e)\subseteq\Upsilon^{new}$.  
\end{proof}
%----------------------------------------------------------------------------------------

Next, for each edge $e=(u,v)\in H$, we present an efficient way to enumerate all bicliques in $\bicliques'(e)$ through enumerating maximal bicliques in a specific subgraph $G'_e$ of $G'$, constructed as follows. Let $A = \Gamma_{G'}(u)$ and $B = \Gamma_{G'}(v)$. Then $G'_e = (A,B,E')$ is a subgraph of $G'$ induced by vertices in $A$ and $B$. See Fig.~\ref{fig:construction-2} for an example of the construction of $G'_e$.

%-------------------------------------------------------------------------------------
\begin{figure}[!t]
\begin{tabular}{c}
\hspace{-4mm}
\includegraphics[width=.5\textwidth]{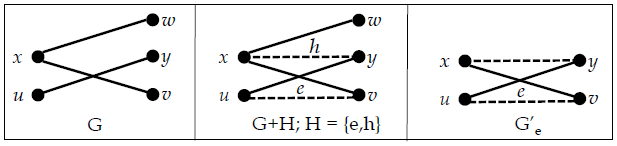}\\
\end{tabular}
\caption{Construction of $G'_e$ from $G' = G+H$ when a set of new edges $H=\{e,h\}$ is added to $G$. $A=\Gamma_{G'}(v)=\{u,x\}$ and $B=\Gamma_{G'}(u)=\{v,y\}$.}
\label{fig:construction-2}
%fig7
\end{figure}
%-------------------------------------------------------------------------------------

%--------------------------------------------------------------------------------------------------
\begin{lemma}\label{lem:new2}
For each $e\in H$, $\bicliques'(e) = \bicliques(G'_e)$
\end{lemma}
%--------------------------------------------------------------------------------------------------

%--------------------------------------------------------------------------------------------------
\begin{proof}
First we show that $\bicliques'(e)\subseteq \bicliques(G'_e)$. Consider a biclique $b = (X,Y)$ in $\bicliques'(e)$. Let $e = (u,v)$. Here $b$ contains both $u$ and $v$. Suppose that $u\in X$ and $v\in Y$. According to the construction $G'_e$ contains all the vertices adjacent to $u$ and all the vertices adjacent to $v$. And in $b$, all the vertices in $X$ are connected to all the vertices in $Y$. Hence, $b$ is a biclique in $G'_e$. Also, $b$ is a maximal biclique in $G'$, and $G'_e$ is an induced subgraph of $G'$ which contains all the vertices of $b$. Hence, $b$ is a maximal biclique in $G'_e$.

Next we show that $\bicliques(G'_e)\subseteq \bicliques'(e)$. Consider a biclique $b' = (X',Y')$ in $\bicliques(G'_e)$. Clearly, $b'$ contains $e$ as it contains both $u$ and $v$ and $b'$ is a maximal biclique in $G'_e$. Hence, $b'$ is also a biclique in $G'$ that contains $e$. Now we prove that $b'$ is also maximal in $G'$. Suppose not, that there is a vertex $w\in V(G')$ such that $b'$ can be extended with $w$. Then, as per the construction of $G'_e$, $w\in V(G'_e)$ since $w$ must be adjacent to either $u$ or $v$. Then, $b'$ is not maximal in $G'_e$. This is a contradiction. Hence, $b'$ is also maximal in $G'$. Therefore, $b'\in\bicliques'(e)$. 
\end{proof}
%--------------------------------------------------------------------------------------------------

Based on the above observation, we present our change-sensitive algorithm $\csnewb$ (Algorithm~\ref{algo:newbc1}). We use an output-sensitive algorithm for a static graph $\mbe$ for enumerating maximal bicliques from $G'_e$. Note that typically, $G'_e$ is much smaller than $G'$ since it is localized to edge $e$, and hence enumerating all maximal bicliques from $G'_e$ should be relatively inexpensive.

%---------------------------- 
\begin{algorithm}
\DontPrintSemicolon
\caption{$\csnewb(G,H)$}
\label{algo:newbc1}
\KwIn{$G$ - Input bipartite graph, $H$ - Edges being added to $G$}
\KwOut{bicliques in $\Upsilon^{new}$, each biclique output once}
Consider edges of $H$ in an arbitrary order $e_1, e_2, \ldots, e_{\rho}$\;
$G'\gets G + H$\;
\For{$i = 1\ldots \rho$}{
	$e \gets e_i=(u,v)$\;
	$G'_e \gets$ a subgraph of $G'$ induced by $\Gamma_{G'}(u)\cup\Gamma_{G'}(v)$\;
	Generate bicliques of $G'_e$ using $\mbe$. For each biclique thus generated,
	output b only if b does not contain an edge $e_j$ for $j < i$\;
}
\end{algorithm}
%----------------------------
%-----------------------------------------------------------------------------------------------------------------------------------------------
\begin{theorem}\label{thm:newbc-1}
$\csnewb$ enumerates the set of all new bicliques arising from the addition of $H$ in time $O(\Delta^{2}\rho|\Upsilon^{new}|)$ where $\Delta$ is the maximum degree of a vertex in $G'$ and $\rho$ is the size of $H$. The space complexity is $O(|E(G')| + \Delta^{2})$.
\end{theorem} 

\begin{proof}
First we consider correctness of the algorithm. From Lemma~\ref{lem:new1} and Lemma~\ref{lem:new2}, we know that $\Upsilon^{new}$ is enumerated by enumerating $\bicliques(G'_e)$ for every $e\in H$. Our algorithm does this exactly, and use the $\mbe$ algorithm for enumerating $\bicliques(G'_e)$.

For the runtime, consider that the algorithm iterates over each edge $e$ in $H$. In each iteration, it constructs a graph $G'_e$ and runs $\mbe(G'_e)$. Note that the number of vertices in $G'_e$ is no more than $2\Delta$, since it is the size of the union of the edge neighborhoods of $\rho$ edges in $G'$. The set of maximal bicliques generated in each iteration is a subset of $\Upsilon^{new}$, therefore the number of maximal bicliques generated from each iteration is no more than $|\Upsilon^{new}|$. From Theorem~\ref{thm2}, we have that the runtime of each iteration is $O(\Delta^{2}|\Upsilon^{new}|)$. Since there are $\rho$ edges in $H$, the result on runtime follows. For the space complexity, we note that the algorithm does not store the set of new bicliques in memory at any point. The space required to construct $G'_e$ is linear in the size of $G'$. From Theorem~\ref{thm2}, the total space requirement is $O(|E(G')| + \Delta^{2})$.
\end{proof}
%-----------------------------------------------------------------------------------------------------------------------------------------------

%----------------------------------------------------------------------------------------
\subsection{Enumerating Subsumed Maximal Bicliques}
\label{sub_bicliques}
%----------------------------------------------------------------------------------------

We now present a change-sensitive algorithm for enumerating $\bicliques(G)\setminus\bicliques(G')$ where $G' = G+H$.
Suppose a new maximal biclique $b$ of $G'$ subsumed a maximal biclique $b'$ of $G$. Note that $b'$ is also a maximal biclique in $b-H$. So, one idea is to enumerate all maximal bicliques in $b-H$ and then check which among them is maximal in $G$. However, checking maximality of a biclique is costly operation since we need to consider neighborhood of every vertex in the biclique. Another idea is to store the bicliques of the graph explicitly and see which among the generated bicliques are contained in the set of maximal bicliques of $G$. This is not desirable either since large amount of memory is required to store the set of all maximal bicliques of $G$.

A more efficient approach is to store the signatures of the maximal bicliques instead of storing the bicliques themselves. Then, we enumerate all maximal bicliques in $b-H$ and for each biclique generated, we compare the signature of the generated biclique with the signatures of the bicliques stored. An algorithm following this idea is presented in Algorithm~\ref{algo:cssubb1}. In this algorithm we reduce the cost of main memory by storing the signatures. We use a standard hash function (such as $64$ bit murmur hash~\footnote{https://sites.google.com/site/murmurhash/}) for computing signatures of maximal bicliques. For computing the signature, first we represent a biclique in canonical form (vertices in first partition represented in lexicographic order followed by vertices in another partition represented in lexicographic order). Then we convert the string into bytes, and apply hash function on the computed bytes. The hash function returns signature as output. By storing the signatures instead of maximal bicliques, we are able to check whether a maximal biclique from $b-H$ is contained in the set of maximal bicliques of $G$ by comparing their hash values. Thus we pay much less cost in terms of memory by storing the signatures of bicliques.

%In Algorithm~\ref{algo:cssubb1}, computation of all maximal bicliques of $b-H$ is performed in Lines $4$ to $12$ followed by checking which among the generated bicliques are maximal in $G$ in Lines $13$ to $15$.

Now we prove that Algorithm~\ref{algo:cssubb1} indeed enumerates all maximal bicliques of $b-H$.

%--------------------------------------------------------------------
\begin{lemma}\label{lem:sub1}
In Algorithm~\ref{algo:cssubb1}, for each $b\in\Upsilon^{new}$, $S$ after Line $14$ contains all maximal bicliques in $b-H$.
\end{lemma}

\begin{proof}
First observe that, removing $H$ from $b$ is equivalent to removing those edges in $H$ which are present in $b$. Hence, computing maximal bicliques in $b-H$ reduces to computing maximal bicliques in $b-H_1$ where $H_1$ is the set of all edges in $H$ which are present in $b$.

We use induction on the number of edges $k$ in $H_1$. Consider the base case, when $k=1$. $H_1$ contains a single edge $e_1=\{u,v\}$. Clearly, $b-H_1$ has two maximal bicliques $b\setminus\{u\}$ and $b\setminus\{v\}$. Suppose, that the set $H_1$ is of size $k$. Our inductive hypothesis is that all maximal bicliques in $b-H_1$ are enumerated. Consider $H_1' = \{e_1, e_2, ..., e_k, e_{k+1}\}$ with $k+1$ edges. Now each maximal biclique $b'$ in $b-H_1$ either remains maximal within $b-H_1'$ (if at least one endpoint of $e_{k+1}$ is not in $b'$) or generates two maximal bicliques in $b-H_1'$ (if both endpoints of $e_{k+1}$ are in $b'$). Thus, for each $b\in\Upsilon^{new}$, $S$ after Line $14$ contains all maximal bicliques within $b-H$.  
\end{proof}
%--------------------------------------------------------------------

%-------------------
\begin{algorithm}[t]
\DontPrintSemicolon
\caption{$\cssubb(G,H,BC,\Upsilon^{new})$}
\label{algo:cssubb1}
\KwIn{$G$ - Input bipartite graph \\ \hspace{1cm} $H$ - Edge set being added to $G$ \\ \hspace{1cm} $BC$ - Set of maximal bicliques in $G$ \\ \hspace{1cm} $\Upsilon^{new}$ - set of new maximal bicliques in $G + H$}
\KwOut{All cliques in $\Upsilon^{del} = \bicliques(G) \setminus \bicliques(G+H)$}
$\Upsilon^{del} \gets \emptyset$\;
\For{$b \in \Upsilon^{new}$}{
		$S\gets \{b\}$\;
		\For{$e = (u,v) \in E(b) \cap H$}{
			$S'\gets \phi$\;
			\For{$b' \in S$}{
				\If{$e \in E(b')$}{
					$b_1 = b'\setminus\{u\}$ ; $b_2 = b'\setminus\{v\}$\;
					$S' \gets S'\cup b_1$ ; $S' \gets S'\cup b_2$\;
					
				}
				\Else{
					
					$S' \gets S'\cup b'$\;
				}
			}
			\tcc{$S'$ contains all the maximal bicliques in $b-\{e_1, e_2, ..., e_k\}$ where $\{e_1, e_2, ..., e_k\}\subseteq E(b) \cap H$ are considered so far.}\;
			$S\gets S'$\;
		}
	\For{$b' \in S$}{
		\If{$b' \in BC$}{
			$\Upsilon^{del} \gets \Upsilon^{del} \cup b'$\;
		}
	}
}
\end{algorithm}
%-------------------

Now we show that the algorithm described above is a change-sensitive algorithm for enumerating all elements of $\Upsilon^{del}$ when the number of edges $\rho$ in $H$ is constant.

%-----------------------------------------------------------------------------------------------------------------------------------------------------------------------------------------------------------
\begin{theorem}\label{thm:subbc-1}
Algorithm~\ref{algo:cssubb1} enumerates all bicliques in $\Upsilon^{del}=\bicliques(G)-\bicliques(G+H)$ using time $O(2^{\rho}|\Upsilon^{new}|)$ where $\rho$ is the number of edges in $H$. The space complexity of the algorithm is $O(|E(G')| + |V(G')| + \Delta^{2} + |\bicliques(G)|)$.
\end{theorem}

\begin{proof}
We first show that every biclique $b'$ enumerated by the algorithm is indeed a biclique in $\Upsilon^{del}$. Note that $b'$ is a maximal biclique in $G$, due to explicitly checking the condition. Further, $b'$ is not a maximal biclique in $G+H$, since it is a proper subgraph of $b$, a maximal biclique in $G+H$. Next, we show that all bicliques in $\Lambda^{del}$ are enumerated. Consider any subsumed biclique $b'\in\Lambda^{del}$. It must be contained within $b\setminus H$,
where $b$ is a maximal biclique within $\Lambda^{new}$. Moreover, $b'$ will be a maximal biclique within $b\setminus H$, and will be enumerated by the algorithm according to Lemma~\ref{lem:sub1}.

For the time complexity we show that for any $b\in\Upsilon^{new}$, the maximum number of maximal bicliques in $b-H$ is $2^{\rho}$ using induction on $\rho$. Suppose $\rho=1$ so that $H$ contains a single edge, say $e_1=(u,v)$. Then, $b-H$ has two maximal bicliques, $b\setminus\{u\}$ and $b\setminus\{v\}$, proving the base case. Suppose that for any set $H$ of size $k$, it was true that $b-H$ has no more than $2^{k}$ maximal bicliques. Consider a set $H'' =
\{e_1,e_2,\ldots,e_{k+1}\}$ with $k+1$ edges. Let $H'=\{e_1,e_2,\ldots,e_k\}$. Subgraph $b-H''$ is obtained from $b-H'$ by deleting a single edge $e_{k+1}$. By induction, we have that $b-H'$ has no more than $2^{k}$ maximal bicliques. Each maximal biclique $b'$ in $b-H'$ either remains a maximal biclique within $b-H''$ (if at least one endpoint of $e_{k+1}$ is not in $b'$),
or leads to two maximal bicliques in $b-H''$(if endpoints of $e_{k+1}$ are in different bipartition of $b'$). Hence, the number of maximal bicliques in $b-H''$ is no more than $2^{k+1}$, completing the inductive step. 

Following this, for each biclique $b\in\Upsilon^{new}$, we need to check for maximality for no more than $2^{\rho}$ bicliques in $G$. This checking can be performed by checking whether each such generated biclique in contained in the set $\bicliques(G)$ and for each biclique, this can be done in constant time. 

For the space bound, we first note that in Algorithm~\ref{algo:cssubb1}, enumerating maximal bicliques within $b-H$ consumes space $O(|E(G')| + \Delta^2)$, and checking for maximality can be done in space linear in size of $G$. However, for storing the maximal bicliques in $G$ takes $O(|\bicliques(G)|)$ space. Hence, for these operations, the overall space-cost for each $b\in\Upsilon^{new}$ is $O(|E(G')| + |V(G')| + \Delta^2 + |\bicliques(G)|)$. The only remaining space cost is the size of $\Upsilon^{new}$, which can be large. Note that, the algorithm only iterates through $\Upsilon^{new}$ in a single pass. If elements of $\Upsilon^{new}$ are provided as a stream from the output of an algorithm such as $\csnewb$, then they do not need to be stored within a container, so that the memory cost of receiving $\Upsilon^{new}$ is reduced to the cost of storing a single maximal biclique within $\Upsilon^{new}$ at a time.
\end{proof}
%-----------------------------------------------------------------------------------------------------------------------------------------------------------------------------------------------------------

%---------------------------- 
\begin{algorithm}
\DontPrintSemicolon
\caption{Decremental(G,H)}
\label{algo:dec}
\KwIn{$G$ - Input bipartite graph, $H$ - Edges being deleted from $G$}
\KwOut{$\Upsilon^{new}(G,G-H)\cup\Upsilon^{del}(G,G-H)$}
$\Upsilon^{new}\gets\phi$; $\Upsilon^{del}\gets\phi$; $G''\gets G-H$\;
$\Upsilon^{del}\gets \csnewb(G'',H)$\;
$\Upsilon^{new}\gets \cssubb(G'', H, \bicliques(G''), \Upsilon^{del})$\;
$\Return \Upsilon^{new}\cup\Upsilon^{del}$\;
\end{algorithm}
%----------------------------

%----------------------------------------------------------------------------------------
\subsection{Decremental and Fully Dynamic Cases}
\label{decremental}
%----------------------------------------------------------------------------------------
We now consider the maintenance of maximal bicliques in the decremental case, when edges are deleted from the graph. This case can be handled using a reduction to the incremental case. We show that the maintenance of maximal bicliques due to deletion of a set of edges $H$ from a bipartite graph $G$ is equivalent to the maintenance of maximal bicliques due to addition of $H$ to the bipartite graph $G-H$.

%-----------------
\begin{lemma}
$\Upsilon^{new}(G,G-H) = \Upsilon^{del}(G-H,G)$ and $\Upsilon^{del}(G,G-H) = \Upsilon^{new}(G-H, G)$
\end{lemma}
%---------------

%-----------------
\begin{proof}
%We show that $\Upsilon^{new}(G,G-H) = \Upsilon^{del}(G-H,G)$ by showing $\Upsilon^{new}(G,G-H)\subseteq \Upsilon^{del}(G-H,G)$ and $\Upsilon^{del}(G-H,G)\subseteq \Upsilon^{new}(G,G-H)$.

Note that $\Upsilon^{new}(G,G-H)$ is the set of all bicliques that are maximal in $G-H$, but not in $G$. By definition, this is equal to $\Upsilon^{del}(G-H,G)$.
Similarly we can show that $\Upsilon^{del}(G,G-H) = \Upsilon^{new}(G-H, G)$.

%First we show that $\Upsilon^{new}(G,G-H) \subseteq \Upsilon^{del}(G-H,G)$. Consider a biclique $b \in\Upsilon^{new}(G,G-H)$. This means, $b\in\bicliques(G-H)$ and $b\notin\bicliques(G)$. This is exactly same as saying $b\in\Upsilon^{del}(G-H,G)$.  Next we show that $\Upsilon^{del}(G-H,G)\subseteq \Upsilon^{new}(G,G-H)$. Consider a biclique $b'\in \Upsilon^{del}(G-H,G)$. This means, $b'\in\bicliques(G-H)$ and $b'\notin\bicliques(G)$, which . This is same as saying $b'\in \Upsilon^{new}(G,G-H)$.
\end{proof}
%-----------------

Based on the above lemma, an algorithm for the decremental case is presented in Algorithm~\ref{algo:dec}.  For the fully dynamic case, where we need to consider both the addition and deletion of edges, we first compute the changes due to addition of edges, followed by changes due to deletion of edges.
%---------------------------------------------------------
\section{{Magnitude of change in Bicliques}}
\label{moc}

We consider the maximum change in the set of maximal bicliques when a set of edges is added to the bipartite graph. Let $\lambda(n)$ denote the maximum size of $\Upsilon(G,G+H)$ taken over all $n$ vertex bipartite graphs $G$ and edge sets $H$. We derive the following upper bound on the maximum size of $\Upsilon(G,G+H)$ in the following Lemma:

\begin{lemma}\label{upper_bound}
$\lambda(n)\leq 2\maxbicliques(n)$.
\end{lemma}

\begin{proof}
Note that, for any bipartite graph $G$ with $n$ vertices and for any new edge set $H$ it must be true that $|\bicliques(G)|\leq \maxbicliques(n)$ and $|\bicliques(G+H)|\leq \maxbicliques(n)$. Since $|\Upsilon^{new}(G,G+H)|\leq |\bicliques(G+H)|$ and $|\Upsilon^{del}(G,G+H)|\leq |\bicliques(G)|$, it follows that $|\Upsilon(G,G+H)| \leq |\bicliques(G+H)| + |\bicliques(G)| \leq 2\maxbicliques(n)$.
\end{proof}

Next we analyze the upper bound of $|\Upsilon(G,G+e)|$ in the following when an edge $e\notin E(G)$ is added to $G$.

%We compute the exact expression of $|\Upsilon(G,G+e)|$ in the following when an edge $e\notin E(G)$ is added to $G$.

\begin{theorem}\label{soc:theorem1}
For an integer $n\geq 2$, a bipartite graph $G = (L, R, E)$ with $n$ vertices, and any edge $e = (u,v)\notin E(G), u\in U, v\in V$, the maximum size of $\Upsilon(G,G+e)$ is $3\maxbicliques(n-2)$, and for each even $n$, there exists a bipartite graph that achieves this bound.
\end{theorem}

We prove this theorem in the following two lemmas. In Lemma~\ref{lem:l1} we prove that the size of $\Upsilon(G,G+e)$ can be as large as $3\maxbicliques(n-2)$  in Lemma~\ref{lem:l4} we prove that the size of $\Upsilon(G,G+e)$ is at most $3\maxbicliques(n-2)$.

%------------------------------------------------------------------------------------------------------------------------------------------------------------------------------
\begin{lemma}\label{lem:l1}
For any even integer $n>2$ there exists a bipartite graph $G$ on $n$ vertices and an edge $e=(u,v)\notin E(G)$ such that $|\Upsilon(G,G+e)|=3\maxbicliques(n-2)$.
\end{lemma}

\begin{proof}
We use proof by construction. Consider bipartite graph $G=(L, R, E)$ constructed on vertex set $U\cup V$ with $n$ vertices such that $|L| = |R| = n/2$. Let $u\in L$ and $v\in R$ be two vertices and let $L'= L\setminus\{u\}$ and $R'=R\setminus\{v\}$. Let $G''$ denote the induced subgraph of $G$ on vertex sets $L'$ and $R'$. In our construction, $G''$ is $CP(\frac{n}{2}-1)$. In graph $G$, in addition to the edges in $G''$, we add an edge from each vertex
in $R'$ to $u$ and an edge from each vertex in $L'$ to $v$. We add edge $e= (u,v)$ to $G$ to get graph $G' = G+e$ (see Fig.~\ref{fig:size-of-change-2} for construction). We claim that the size of $\Upsilon(G,G')$ is $3\maxbicliques(n-2)$.

First, we note that the total number of maximal bicliques in $G$ is $2\maxbicliques(n-2)$. Each maximal biclique in $G$ contains either vertex $u$ or $v$, but not both. The number of maximal bicliques that contain vertex $u$ is $\maxbicliques(n-2)$, since each maximal biclique in $G''$ leads to a maximal biclique in $G$ by adding $u$. Similarly, the number of maximal bicliques in $G$ that contains $v$ is $\maxbicliques(n-2)$, leading to a total of $2\maxbicliques(n-2)$ maximal bicliques in $G$.

Next, we note that the total number of maximal bicliques in $G'$ is $\maxbicliques(n-2)$. To see this, note that each maximal biclique in $G'$ contains both vertices $u$ and $v$. Further, for each maximal biclique in $G''$, we get a corresponding maximal biclique in $G'$ by adding vertices $u$ and $v$. Hence the number of maximal bicliques in $G'$ equals the number of maximal bicliques in $G''$, which is $\maxbicliques(n-2)$.

No maximal biclique in $\bicliques(G)$ contains both $u$ and $v$, while every maximal biclique in $G'$ contains both $u$ and $v$. Hence, $\bicliques(G)$ and $\bicliques(G')$ are disjoint sets, and $|\Upsilon(G,G')| = |\bicliques(G)| + |\bicliques(G')| = 3\maxbicliques(n-2)$.
\end{proof}
%----------------------------------------------------------------------------------------------------------------------------------------------------------------------------

%-------------------------------------------------------------------------------------
\begin{figure}[t!]
\centering
\begin{tabular}{c}
%\hspace{-2mm}
%\includegraphics[width=.45\textwidth]{pic2.png}\\
%\includegraphics[width=.6\textwidth]{bpic6.png}\\
\includegraphics[width=2.5in]{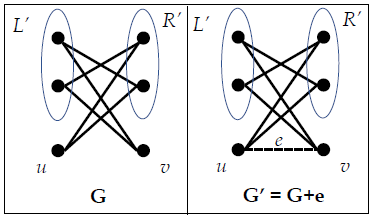}\\
\end{tabular}
\caption{Construction showing the changes in the set of maximal bicliques when a new edge is added. $G$ is in the left on $n = 6$ vertices. $G''$ consists of vertices in $L'$ and $R'$ and edges among them to make it a cocktail-party graph. $G'$ in the right is obtained by adding edge $e = (u,v)$ to $G$.}
\label{fig:size-of-change-2}
%fig7
\end{figure}
%-------------------------------------------------------------------------------------
Now we will prove a few results that we will use in proving Lemma~\ref{lem:l4}.
%-----------------------------------------------------------------------------------------------------------------------------------------------------------------------------
\begin{lemma}\label{lem:l2}
If $e=(u,v)\notin E(G)$ is inserted to $G$ where $u\in L, v\in R$, all new maximal bicliques in $G+e$ must contain $e$. 
\end{lemma}

\begin{proof}
Proof by contradiction. Assume that there is a new maximal biclique $b = (b_1, b_2)$ in $\bicliques(G+e)-\bicliques(G)$ that does not contain $e$. Then $b$ must be present in $G$ but is not maximal in $G$, and there must be another vertex $w\in L$ (or $R$) that can be added to $b$ while remaining a biclique. Clearly, $w$ can be added to biclique $b$ in $G+e$ also, so that $b$ is not maximal in $G+e$, contradicting our assumption.
\end{proof}
%-----------------------------------------------------------------------------------------------------------------------------------------------------------------------------
%-----------------------------------------------------------------------------------------------------------------------------------------------------------------------------
\begin{lemma}\label{lem:l3}
If $e = (u,v)$ is added to $G$, each biclique $b\in\bicliques(G)-\bicliques(G+e)$ contains either $u$ or $v$. 
\end{lemma}

\begin{proof}
Proof by contradiction. Suppose there is maximal biclique $b=(b_1, b_2)$ in $\bicliques(G)-\bicliques(G+e)$ that contain neither $u$ nor $v$. Then, $b$ must be maximal biclique in $G$. Since $b$ is not maximal biclique in $G+e$, $b$ is contained in another maximal biclique $b' = (b_1', b_2')$ in $G+e$. From Lemma~\ref{lem:l2}, $b'$ must contain edge $e=(u,v)$, and hence, both vertices $u$ and $v$. Since $b'$ is a biclique, every vertex in $b_2'$ is connected to $u$ in $G'$. Hence, every vertex in $b_2$ is connected to $u$ even in $G$. Therefore, $b\cup\{u\}$ is a biclique in $G$, and $b$ is not maximal in $G$, contradicting our assumption.
\end{proof}
%-------------------------------------------------------------------------------------------------------------------------------------------------------------------------------

%-------------------------------------------------------------------------------------------------------------------------------------------------------------------------------
\begin{observation}\label{obs:o1}
For a bipartite graph $G=(L, R, E)$ and a vertex $u\in V(G)$, the number of maximal bicliques that contains $v$ is at most $\maxbicliques(n-1)$.
\end{observation}

\begin{proof}
Suppose, $u\in L$. Then each maximal biclique $b$ in $G$ that contains $u$, corresponds to a unique maximal biclique in $G-\{u\}$. Such maximal bicliques can be derived from $b$ by deleting $u$ from $b$. As the maximum number of maximal bicliques in $G-\{u\}$ is $\maxbicliques(n-1)$, maximum number of maximal bicliques in $G$ can be no more than $\maxbicliques(n-1)$.
\end{proof}
%-------------------------------------------------------------------------------------------------------------------------------------------------------------------------------
%-------------------------------------------------------------------------------------------------------------------------------------------------------------------------------
\begin{observation}\label{obs:o2}
The number of maximal bicliques containing a specific edge $(u,v)$ is at most $\maxbicliques(n-2)$.
\end{observation}

\begin{proof}
Consider an edge $(u,v)\in E(G)$. Let vertex set $V' = (\Gamma_G(u)\cup\Gamma_G(v)) - \{u,v\}$, and let $G'$ be the subgraph of $G$ induced by $V'$. Each maximal biclique $b$ in $G$ that contains edge $(u,v)$ corresponds to a unique maximal biclique in $G'$ by simply deleting vertices $u$ and $v$ from $b$. Also, each maximal biclique $b'$ in $G'$ corresponds to a unique maximal biclique in $G$ that contains $(u,v)$ by adding vertices $u$ and $v$ to $b'$. Thus, there is a bijection between the
maximal bicliques in $G'$ and the set of maximal bicliques in $G$ that contains edge $(u,v)$. The number of maximal bicliques in $G'$ can be at most $\maxbicliques(n-2)$ since $G'$ has no more than $(n-2)$ vertices, completing the proof. 
\end{proof}
%-------------------------------------------------------------------------------------------------------------------------------------------------------------------------------

\begin{lemma}\label{lem:l4}
For a bipartite graph $G=(L, R, E)$ on $n$ vertices and edge $e=(u,v)\notin E(G)$, the size of $\Upsilon(G,G+e)$ can be no larger than $3\maxbicliques(n-2)$.
\end{lemma}

\begin{proof}
Proof by contradiction. Suppose there exists a bipartite graph $G=(L, R, E)$ and edge $e\notin E(G)$ such that $|\Upsilon(G,G+e)|\ge 3\maxbicliques(n-2)$. Then either $|\bicliques(G+e)-\bicliques(G)|\ge \maxbicliques(n-2)$ or $|\bicliques(G)-\bicliques(G+e)|\ge 2\maxbicliques(n-2)$.

\textbf{Case 1:} $|\bicliques(G+e)-\bicliques(G)|\ge \maxbicliques(n-2)$: This means that total number of new maximal bicliques formed due to addition of edge $e$ is larger than $\maxbicliques(n-2)$. From Lemma~\ref{lem:l2}, each new maximal biclique formed due to addition of $e$ must contain $e$. From Observation~\ref{obs:o2}, the total number of maximal bicliques in an $n$ vertex bipartite graph containing a specific edge can be at most $\maxbicliques(n-2)$. Thus, the number of new maximal bicliques after
adding edge $e$ is at most $\maxbicliques(n-2)$, contradicting our assumption.

\textbf{Case 2:} $|\bicliques(G)-\bicliques(G+e)|\ge 2\maxbicliques(n-2)$: Using Lemma~\ref{lem:l3}, each maximal biclique $b\in\bicliques(G)-\bicliques(G+e)$ must contain either $u$ or $v$, but not both. Suppose that $b$ contains $u$ but not $v$. Then, $b$ must be a maximal biclique in $G-v$. Using Observation~\ref{obs:o1}, we see that the number of maximal bicliques in $G-v$ that contains a specific vertex $u$ is no more than $\maxbicliques(n-2)$. In a similar way, the number of
possible maximal bicliques that contain $v$ is at most $\maxbicliques(n-2)$. Therefore, the total number of maximal bicliques in $\bicliques(G)-\bicliques(G+e)$ is at most $2\maxbicliques(n-2)$, contradicting our assumption.
\end{proof}

Combining Lemma~\ref{upper_bound}, Theorem~\ref{soc:theorem1} and using the fact that $3\maxbicliques(n-2) = 1.5\maxbicliques(n)$ for even $n$, we obtain the following when $n$ is even:
%-------------------
\begin{theorem}
$1.5\maxbicliques(n) \leq \lambda(n) \leq 2\maxbicliques(n)$
\end{theorem}
%-------------------

%-----------------------------------------------------------------------------------------
\section{Experimental Evaluation}
\label{expt}
\newcommand{\naive}{{\tt BaselineBC}}
%-----------------------------------------------------------------------------------------

In this section, we present results of an experimental  evaluation of our algorithms. 
%---------------------------------------------------
\subsection{Data}
%---------------------------------------------------
We consider the following real-world bipartite graphs in our experiments. A summary of the datasets is presented in Table~\ref{table-1}.  In the {\tt actor-movie}~\cite{data:am16} graph, vertices consist of actors in one bipartition and movies in another bipartition. There is an edge between an actor and a movie if the actor played in that movie. In the {\tt dblp-author}~\cite{data:da16} graph, vertices consist of authors in one partition and the publications in another partition. Edges connect authors to their publications. In the {\tt epinions-rating}~\cite{data:er16} graph, vertices consist of users in one partition and products in another partition. There is an edge between a user and a product if the user rated the product. Also, the edges have timestamps of their creation. In the {\tt flickr-membership}~\cite{data:fm16} graph, vertices consists of users and groups. There is an edge between a user and a group if that user is a member of that group.

We converted the above graphs into dynamic graphs by creating edge streams as follows: For {\tt actor-movie}, {\tt dblp-author}, and {\tt flickr-membership} we created initial graphs by retaining each edge in the original graph with probability $0.1$ and deleting the rest. Then the deleted edges are added back as an edge stream, until the original graph is reached. We named the initial graphs as {\tt actor-movie-1}, {\tt dblp-author-1}, and {\tt flickr-membership-1}. For the {\tt epinions-rating} graph, we created the initial graph by retaining initial $10\%$ edges of the original graph according to their timestamps, and considered rest of the edges for creating the edge stream in timestamp ordering. We named the initial graph as {\tt epinions-rating-init}. In Table~\ref{table-1}, the number of edges of the initial graph is in the column {\tt Edges(initial)} and the number of edges when we end the experiment is in column {\tt Edges(final)}.

%-------------------------------------------------------------------------------------
\begin{table*}[!t]
\small
\centering
\caption{\textbf{Summary of Graphs Used}}
\label{table-1}
%\begin{adjustwidth}{-10mm}{}
\scalebox{1}{
\begin{tabular}{| l | c | c | c | c | c |}

\toprule
\textbf{Dataset} & Nodes & Edges(initial) &  Edges(final) & Edges(original graph) & Avg. deg.(original graph)\\
\midrule
{\tt actor-movie}-$1$ & $639286$ & $146917$ & $1470404$ & $1470404$ & $6$ \\
\hline
{\tt dblp-author}-$1$ & $6851776$ & $864372$ & $8649016$ & $8649016$ & $3$ \\
\hline
{\tt epinions-rating-init} & $996744$ & $1366832$ & $1631832$ & $13668320$ & $31$\\
\hline
{\tt flickr-membership}-$1$ & $895589$ & $855179$ & $1355179$ & $8545307$ & $35$ \\
\bottomrule
\end{tabular}
}
%\end{adjustwidth}
\end{table*}
%-------------------------------------------------------------------------------------

%--------------------------------------------------------------------------------
\begin{table*}[t!]
\small
\centering
\caption{\textbf{Comparison with Baseline: computation time for adding a single batch of size $100$}} 
\label{table-3}
\scalebox{1}{
\begin{tabular}{|c| c | c|}
\toprule
\textbf{Initial-graph} & $\mbc$ & $\naive$ \\
\midrule
{\tt actor-movie}-$1$ & $30$ ms. & $> 30$ min. \\
\hline
{\tt dblp-author}-$1$ & $20$ ms. & $> 20$ hours \\
\hline
{\tt epinion-rating-init} & $3.5$ sec. & $> 10$ hours\\
\hline
{\tt flickr-membership}-$1$ & $0.5$ sec. & $1$ hour \\ 
\bottomrule
\end{tabular}
}
%\end{adjustwidth}
\end{table*} 
%--------------------------------------------------------------------------------

%-------------------------------------------------------------------------------------------------------------
\begin{figure*}[t!]
%\begin{adjustwidth}{-9mm}{}
\centering
\begin{tabular}{cc}
	\hspace{-8mm}
	\includegraphics[width=0.45\textwidth]{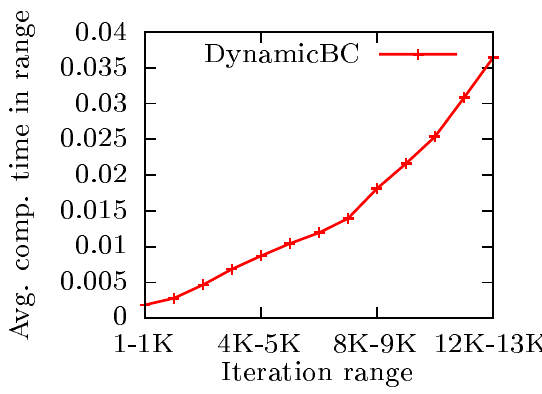} &
       \hspace{-10mm}
	\includegraphics[width=0.45\textwidth]{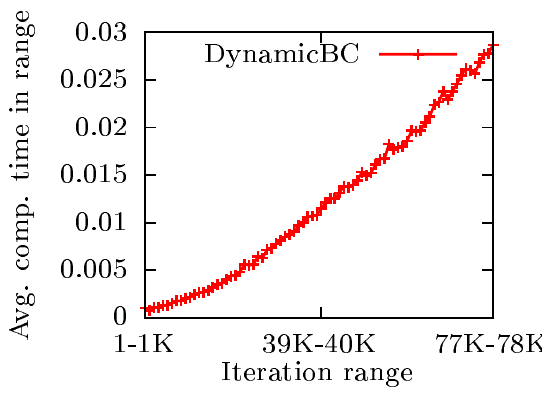} \\
		(a) {\tt actor-movie}-$1$ &
		(b) {\tt dblp-author}-$1$\\
    	%\hspace{-2mm}
	\includegraphics[width=0.45\textwidth]{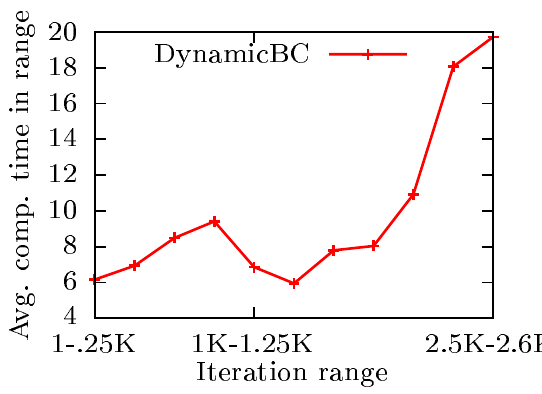} &
	%\hspace{-8mm}
	\includegraphics[width=0.45\textwidth]{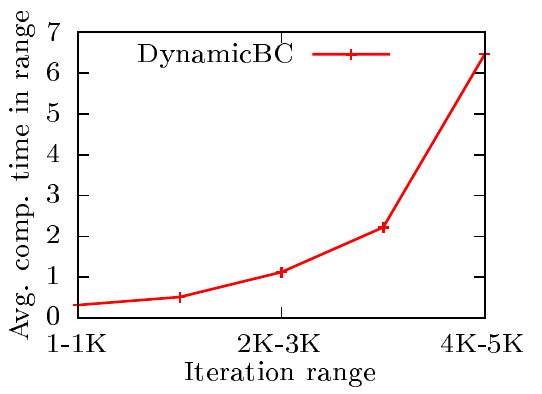} \\
		(c) {\tt epinions-rating-init} &
	%\hspace{-4mm}
		(d) {\tt flickr-membership}-$1$ \\
\end{tabular}
%\end{adjustwidth}
\caption{\textbf{Computation time (in sec.)  for enumerating the change in maximal bicliques, per batch of edges.}}
\label{fig:time_for_total_change}
%fig1
\end{figure*}
%-------------------------------------------------------------------------------------------------------------

%------------------------------
\begin{figure*}[t!]
%\begin{adjustwidth}{-9mm}{}
\centering
\begin{tabular}{cc}
	%\hspace{-16mm}
	\includegraphics[width=.45\textwidth]{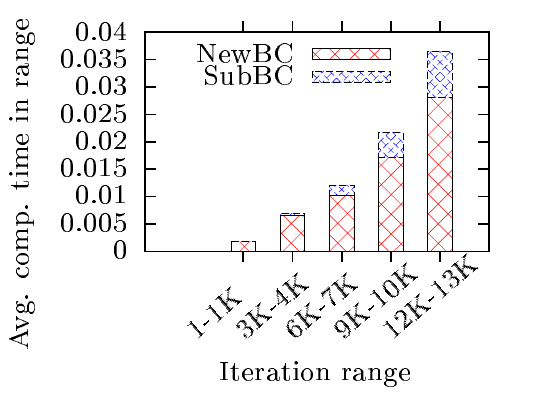} &
       %\hspace{-10mm}
	\includegraphics[width=.45\textwidth]{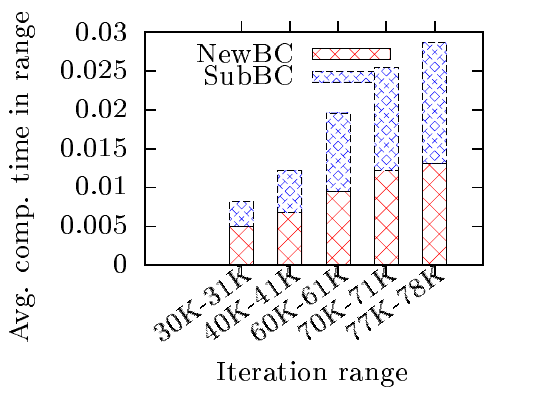} \\
	(a) {\tt actor-movie}-$1$ &
	(b) {\tt dblp-author}-$1$ \\
    	%\hspace{-10mm}
	\includegraphics[width=.45\textwidth]{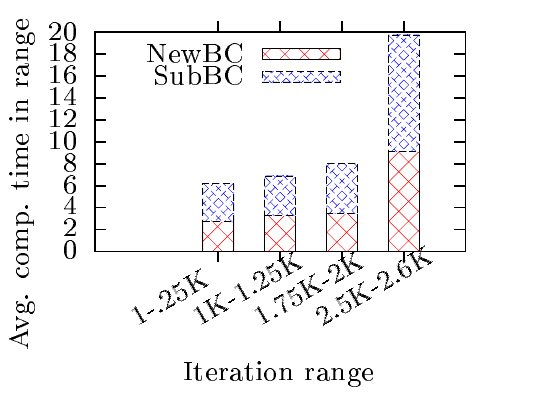} &
	%\hspace{-10mm}
	\includegraphics[width=.45\textwidth]{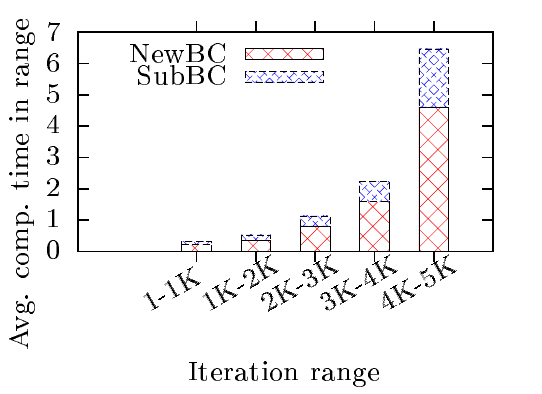} \\
	(c) {\tt epinions-rating-init}&
	(d) {\tt flickr-membership}-$1$ \\
\end{tabular}
%\end{adjustwidth}
\caption{\textbf{Computation  time (in sec.) broken down into time for new and subsumed bicliques}}
\label{fig:time_division}
%fig1
\end{figure*}
%-----------------------------------------

%---------------------------------------------------
\subsection{Experimental Setup and Implementation Details}
%---------------------------------------------------
We implemented our algorithms using Java on a $64$-bit Intel(R) Xeon(R) CPU clocked at $3.10$ Ghz and $8$G DDR3 RAM with $6$G heap memory space. Unless otherwise specified, we considered batches of size $100$.

%We have compared our proposed algorithm with the baseline approach. We name our algorithm for the maintenance of maximal bicliques in the dynamic graph as $\mbc$ that consists of two parts: (1) $\csnewb$ for enumerating the set of new maximal bicliques and (2) $\cssubb$ for enumerating the set of subsumed bicliques.
%enumerates the set of maximal bicliques of $G'$ from scratch once the bipartite graph $G$ is updated to $G'$ by addition of batch $H$ of new edges to $G$ and then computes the symmetric difference of $\bicliques(G)$ and $\bicliques(G')$. 

\textbf{Metrics:} We evaluate our algorithms using the following metrics: (1)~computation time for new maximal bicliques and subsumed bicliques when a set of edges are added, (2)~change-sensitiveness, that is, the total computation time as a function of the size of change. We measure the size of change as the sum of the total number of edges in the new maximal bicliques and the subsumed bicliques, and (3)~space cost, that is the memory used by the algorithm for storing the graph, and other data structures used by the algorithm, and (4)~cumulative computation time for different batch sizes, that is the cumulative computation time from the initial graph to the final graph while using different batch size. 

%---------------------------------------------------------------------------------------------------------------
\begin{figure*}[t!]
%\begin{adjustwidth}{-9mm}{}
\centering
\begin{tabular}{cc}
	%\hspace{-8mm}
	\includegraphics[width=.45\textwidth]{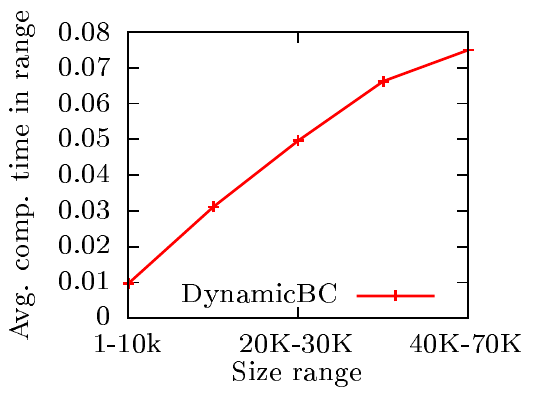} &
       \hspace{-8mm}
	\includegraphics[width=.45\textwidth]{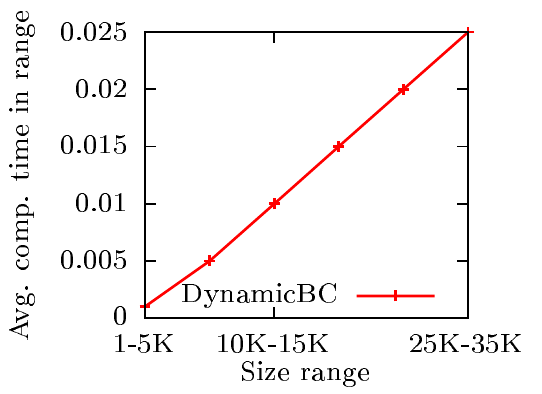} \\
	(a) {\tt actor-movie}-$1$&
	%\hspace{-12mm}
	(b) {\tt dblp-author}-$1$ \\
    	%\hspace{-6mm}
	\includegraphics[width=.45\textwidth]{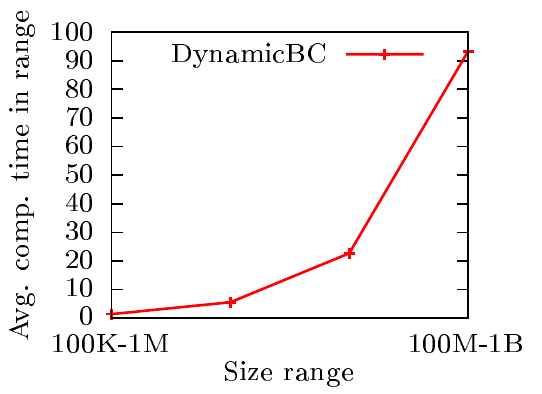} &
	%\hspace{-12mm}
	\includegraphics[width=.45\textwidth]{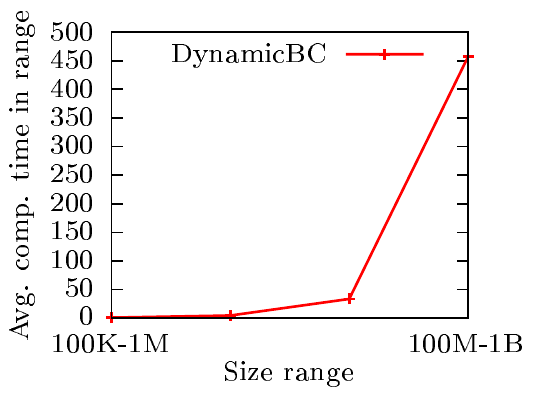} \\
		(c) {\tt epinions-rating-init}&
		\hspace{-4mm}
		(d) {\tt flickr-membership}-$1$ \\
\end{tabular}
%\end{adjustwidth}
\caption{\textbf{Computation time (in sec.) for total change vs. size of total change.}}
\label{fig:change_vs_time}
%fig1
\end{figure*}
%--------------------------------------------------------------------------------------------------------------
%------------------------------------------------------------------------
\subsection{Discussion of Results}
%------------------------------------------------------------------------
{\bf Comparison with Baseline.} We compared the performance of our algorithm, $\mbc$, with a baseline algorithm for maintaining maximal bicliques, we have implemented algorithm that we call $\naive$. The baseline algorithm computes $\Upsilon(G,G+H)$ by (1)~Enumerating $\bicliques(G)$, (2)~Enumerating $\bicliques(G+H)$, and (3)~computing the difference of the two. We use $\mbe$~\cite{LSL06} for enumerating bicliques from a static graph. Table~\ref{table-3} shows a comparison of the runtimes of $\mbc$ and $\naive$.  From the table, it is clear that $\mbc$ is faster than $\naive$ by two to three orders of magnitude. For instance, for adding a single batch of size $100$ to {\tt actor-movie-1}, $\naive$ takes more than $30$ min., whereas $\mbc$ takes around $30$ ms. \\

\textbf{Computation Time per Batch of Edges:} Let an ``iteration" denote the addition of a single batch of edges. Fig.~\ref{fig:time_for_total_change} shows the computation time per iteration versus iteration number. From the plots, we observe that the computation time increases as the iteration increases. This trend is consistent with predictions. Note that as computation progresses, the number of edges in the graph increases, and note that the the computation time is proportional to the size of graph as well as size of change (Theorem~\ref{thm:main}). In Fig.~\ref{fig:time_for_total_change}(c) we see that computation time decreases suddenly and then again increases. This may seem anomalous, but is explained by noting that in these cases, the magnitude of change decreases in those iterations, and then increases thereafter. \\

In Fig.~\ref{fig:time_division}, we show the breakdown of the computation time of $\mbc$ into time taken for enumerating new cliques ($\csnewb$) and for enumerating subsumed cliques ($\cssubb$). Observe that the computation time increases for both new maximal bicliques and subsumed bicliques as more batches are added. This is because the graph becomes denser when more batches are added and the time taken to compute the change increases, consistent with Theorem~\ref{thm:main}.\\

%In most of the cases, computation time of $\csnewb$ and $\cssubb$ increases as more edges are added.

\textbf{Change-Sensitiveness:}  Fig.~\ref{fig:change_vs_time} shows the computation time as a function of the size of change. We observe that the computation time of $\mbc$ is roughly proportional to the size of change. The computation time of both $\csnewb$ and $\cssubb$ increases as number of new maximal bicliques and subsumed bicliques increases. Clearly, this observation supports our theoretical analysis. In some plots (Fig.~\ref{fig:change_vs_time}(c),\ref{fig:change_vs_time}(d)) we see a rapid increase in the computation time with the size of change. This is because, when the graph grows, memory consumption increases considerably and this affects the computation time of the algorithm.\\

\textbf{Space Cost:} Fig.~\ref{fig:space_cost} shows the space cost of $\mbc$ for different graphs. As $\cssubb$ needs to maintain the maximal bicliques in memory for computing subsumed bicliques, we report the space consumption in two cases: (1) when we store the maximal bicliques in memory, (2) when we store the signatures of bicliques in memory instead of storing the bicliques. Signatures consume less memory than the actual bicliques as the signatures have fixed size ($64$ bits in our case using the murmur hash function) for different sizes of bicliques. Therefore, memory consumption by the algorithm that uses signatures should be smaller than the algorithm that does not use signatures. The trend is also clear in the plots. The difference in memory consumption is not prominent during the initial iterations because, sizes of maximal bicliques are much smaller during initial iterations and therefore memory consumption is mainly due to the graph that we maintain in memory. We are not showing the space cost without hash for the third input graph because the algorithm could not execute on the third input graph without hashing, due to running out of memory.\\

%--------------------------------------------------------------------------------------------------------------
\begin{figure*}[t!]
%\begin{adjustwidth}{-9mm}{}
\centering
\begin{tabular}{cc}
	%\hspace{-8mm}
	\includegraphics[width=.45\textwidth]{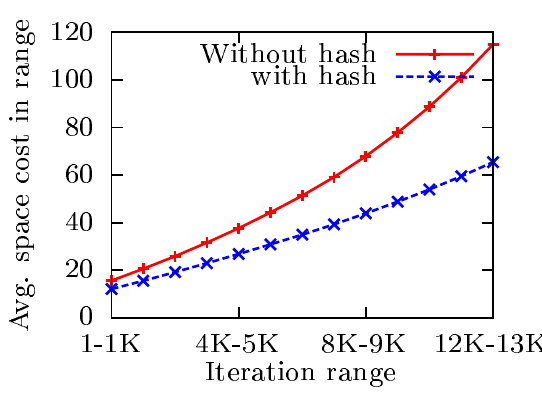} &
       %\hspace{-12mm}
	\includegraphics[width=.45\textwidth]{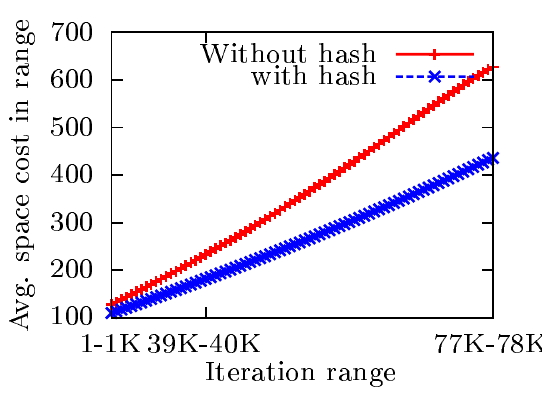} \\
	(a) {\tt actor-movie}-$1$ &
	(b) {\tt dblp-author}-$1$ \\
	%\hspace{-6mm}
	\includegraphics[width=.45\textwidth]{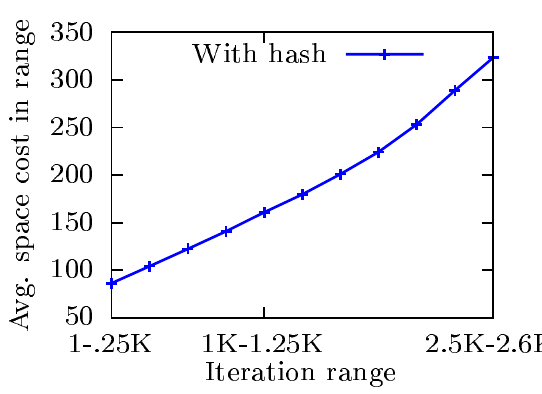} &
    	%\hspace{-12mm}
	\includegraphics[width=.45\textwidth]{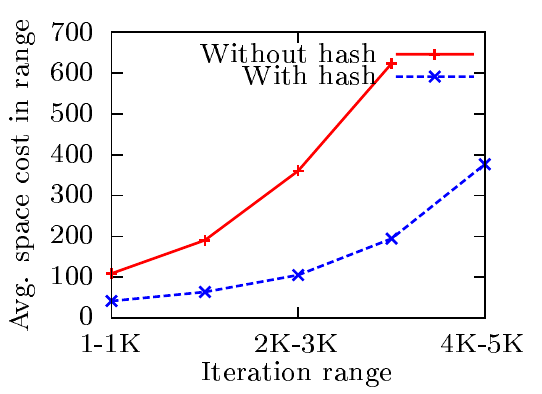} \\
		(c) {\tt epinion-rating-init} &
		(d) {\tt flickr-membership}-$1$ \\
\end{tabular}
%\end{adjustwidth}
\caption{\textbf{Space cost (in MB)}}
\label{fig:space_cost}
%fig1
\end{figure*}
%-------------------------------------------------------------------------------------------------------------

\textbf{Computation Time for Different Batch Size:} Table~\ref{table-2} shows the cumulative computation time for different graphs when we use different batch size. We observe that the total computation time increases when increasing the batch size. The reason for this trend is that the computation time for subsumed cliques increases with increasing batch size, while the computation time for the new maximal bicliques remains almost same across different batch sizes. Note that, the time complexity for $\cssubb$ has (in the worst case) an exponential dependence on the batch size. Therefore, the computation time for subsumed cliques tends to increase with an increase in the batch size. However, with a very small batch size (such as $1$ or $10$), the change in the maximal bicliques is very small, and the overhead can be large.

%we considered large batch size because, with small batch size ($1$ or $10$), the size of changes are very small and therefore showing the change-sensitiveness of the algorithm is not possible. But for practical purposes, choosing lower batch size is useful.\\

%--------------------------------------------------------------------------------
\begin{table*}[!t]
\small
\centering
\caption{\textbf{Total computation time (from the initial graph to the final graph) for different batch sizes}} 
\label{table-2}
%\begin{adjustwidth}{-10mm}{}
\scalebox{1}{
\begin{tabular}{|l |c| c| c|}

\toprule
\textbf{Initial-graph} & \textbf{batch-size}-$1$ & \textbf{batch-size}-$10$ & \textbf{batch-size}-$100$\\
\midrule
{\tt actor-movie}-$1$ & $3.8$ min. ($3.3 + 0.5$) & $3.8$ min. ($2.8 + 1$) & $3.9$ min. ($2.9 + 1$) \\
\hline
{\tt dblp-author}-$1$ & $11.3$ min. (9 + 2.3) & $14.1$ min. ($8.8 + 5.3$) & $15.7$ min. ($8.3 + 7.4$) \\
\hline
{\tt epinion-rating-init} & $3.3$ hours ($3.1 + .2$) & $3.7$ hours ($3.1 + 0.6$) & $7$ hours ($3.2 + 3.8$)\\
\hline
{\tt flickr-membership}-$1$ & $2.1$ hours ($1.9 + 0.2$) & $2.4$ hours ($1.9 + 0.5$) & $3$ hours ($2.1 + 0.9$) \\
\bottomrule
\end{tabular}
}
%\end{adjustwidth}
\end{table*} 
%-----------------------------------------------------------------------------------------
\
%For the first two graph, we observe that the cumulative computation time decreases as we increase the batch size. This is as expected because, when we increase batch size, we reduce computation cost for many bicliques which no longer remain maximal as the graph becomes updated. But for the third input graph,we observe an unusual trend. This is not expected. However, we observe that the computation time for subsumed cliques increases significantly as we increase the batch size, while the computation time for new bicliques remains almost unchanged. 

\textbf{Maintaining Large Maximal Bicliques: } We also consider maintaining large maximal bicliques with predefined size threshold $s$, where it is required that each bipartition of the biclique has size at least $s$.  For large subsumed bicliques, we provide $s$ in addition to other inputs to $\cssubb$ as well. Table~\ref{table-4} shows the cumulative computation time by varying the threshold size $s$ from $1$ to $6$. Clearly, $s=1$ means that we maintain all maximal bicliques. As expected, the cumulative computation time decreases significantly in most of the cases as the size threshold $s$ increases.\\
%For large new maximal bicliques with threshold $s$, we provide input $s$ to $\mbc$ in addition to other inputs.  

%-----------------------------------------------------------------------------------------
\begin{table*}[!t]
\small
\centering
\caption{\textbf{Total computation time (from initial to final graph) by varying the threshold size $s$}}
\label{table-4}
\scalebox{1}{
\begin{tabular}{|l|c|c|c|c|c|c|}
\toprule
\textbf{Initial-graph} & $s=1$ & $s=2$ & $s=3$ & $s=4$ & $s=5$ & $s=6$\\
\midrule
{\tt actor-movie}-$1$ & $203$ sec. & $124$ sec. & $105$ sec. & $100$ sec. & $103$ sec. & $98$ sec. \\
\hline
{\tt dblp-author}-$1$ & $947$ sec. & $531$ sec. & $445$ sec. & $403$ sec. & $399$ sec. & $400$ sec. \\
\hline
{\tt epinion-rating-init} & $7$ hours & $6.5$ hours & $6.3$ hours & $6$ hours & $5.5$ hours & $5$ hours \\
\hline
{\tt flickr-membership}-$1$ & $3$ hours & $2.5$ hours & $2.3$ hours & $2.1$ hours & $1.9$ hours & $1.6$ hours \\
\bottomrule
\end{tabular}
}
\end{table*}

\section{Conclusion}\label{conclude}
In this work, we presented a change-sensitive algorithm for enumerating changes in the set of maximal bicliques in dynamic graph. The performance of this algorithm is proportional to the magnitude of change in the set of maximal bicliques -- when the change is small, the algorithm runs faster, and when the change is large, it takes a proportionally longer time. We present near-tight bounds on the maximum possible change in the set of maximal bicliques, due to a change in the set of edges in the graph.  Our experimental evaluation shows that the algorithm is efficient in practice, and scales to graphs with millions of edges. This work leads to natural open questions (1)~Can we design more efficient algorithms for enumerating the change, especially for enumerating subsumed cliques? (2)~Can we parallelize the algorithm for enumerating the change in maximal bicliques?

\bibliographystyle{abbrv}
\bibliography{dmbe}

\begin{IEEEbiography}[{\includegraphics[width=1.0in,height=1.5in,clip,keepaspectratio]{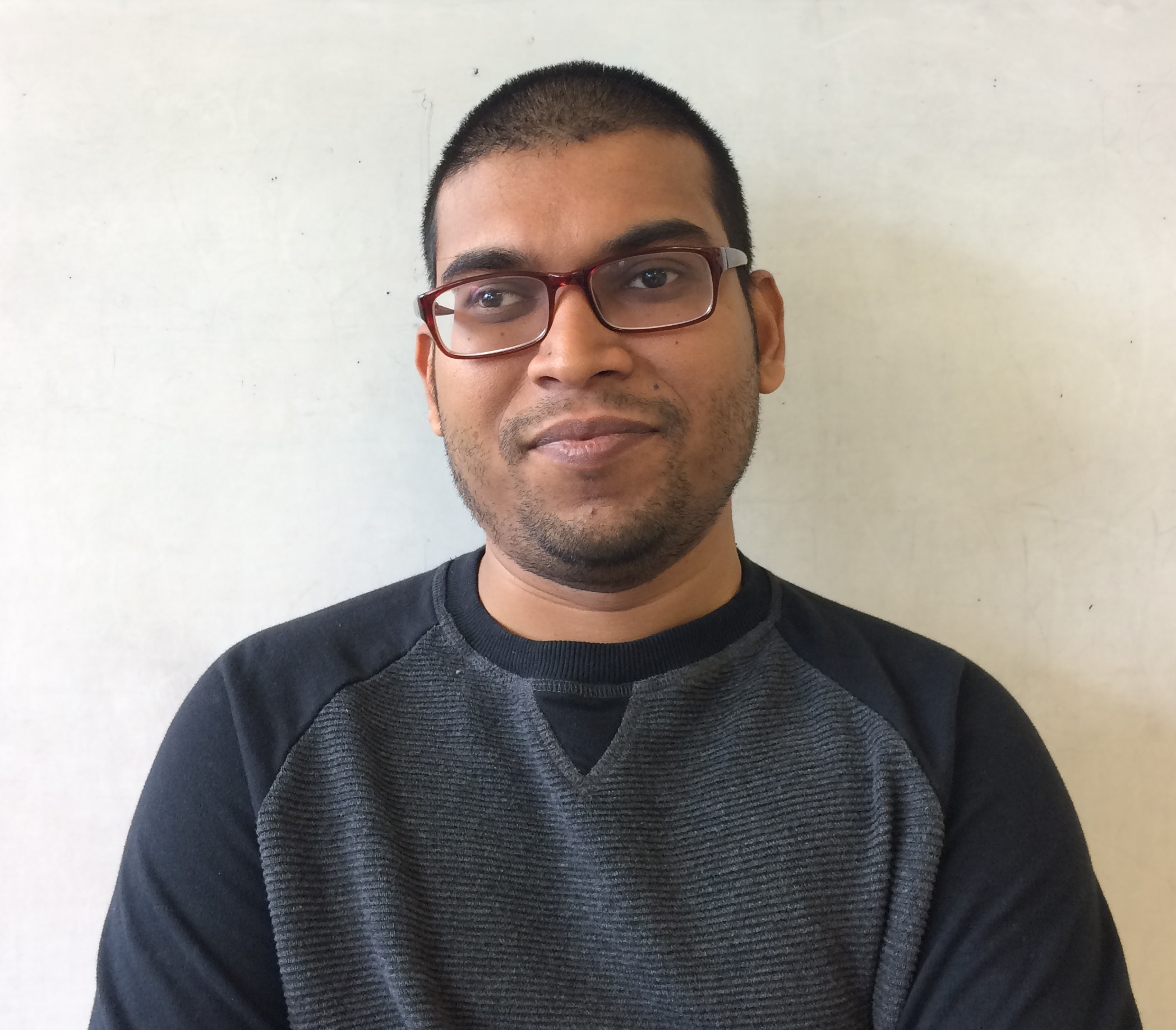}}]{Apurba Das}
is a 4th year Ph.D. student in the department of Computer Engineering at Iowa State University. He received his Masters in Computer Science from Indian Statistical Institute, Kolkata in 2011 and worked for 2 years after that as a software developer at Ixia. His research interests are in the area of graph mining, dynamic and streaming graph algorithms, and large scale data analysis.
\end{IEEEbiography}

\begin{IEEEbiography}[{\includegraphics[width=1.0in,height=1.5in,clip,keepaspectratio]{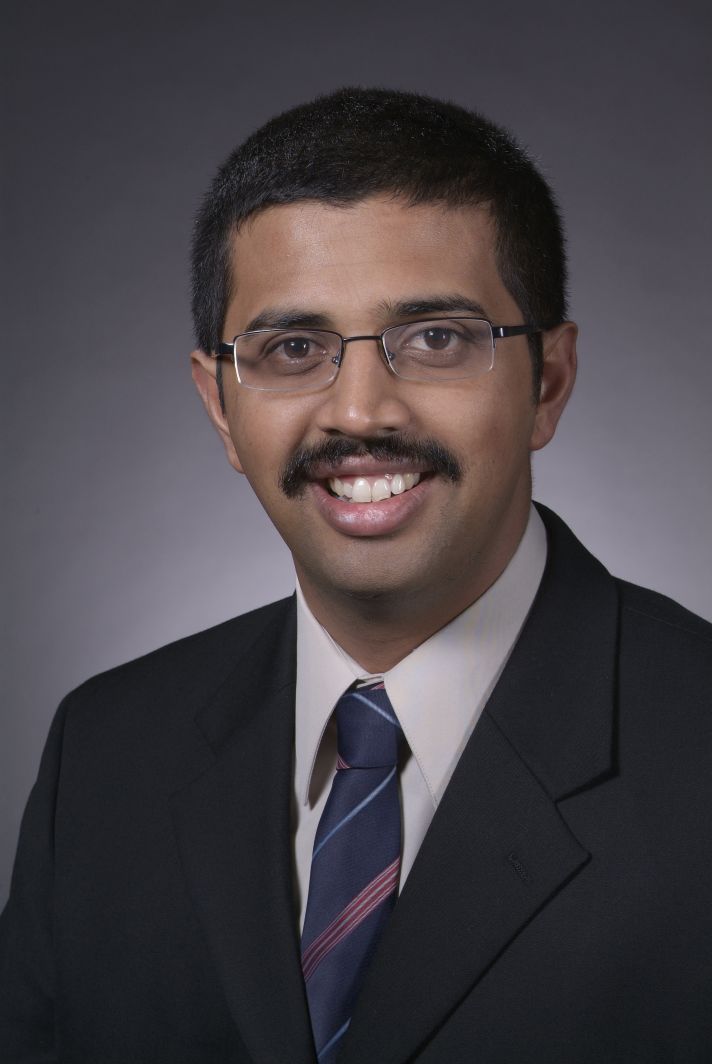}}]{Dr. Srikanta Tirthapura}
received his Ph.D. in Computer Science from Brown University in 2002, and his B.Tech. in Computer Science and Engineering from IIT Madras in 1996. He is the Kingland Professor of Data Analytics in the department of Electrical and Computer Engineering at Iowa State University. He has worked at Oracle Corporation and is a recipient of the IBM Faculty Award, and the Warren Boast Award for excellence in Undergraduate Teaching. His research interests include algorithms for large-scale data analysis, stream computing, and cybersecurity.
\end{IEEEbiography}
\end{document}